\documentclass{article}
\usepackage[final]{corl_2019}
\usepackage[utf8]{inputenc}
\usepackage{amsmath,amssymb}
\usepackage{mathtools}
\usepackage{amsthm}
\usepackage{graphicx}
\usepackage{float}
\usepackage{wrapfig}
\usepackage{subcaption}
\usepackage{subfloat}
\usepackage{algorithmicx,algpseudocode}
\usepackage{algorithm}
\usepackage{amsfonts,dsfont}
\theoremstyle{plain}
\newtheorem{thm}{\textbf{Theorem}}
\newtheorem{lem}{\textbf{Lemma}}

\newtheorem{cor}{\textbf{Corollary}}
\DeclarePairedDelimiter\set{\{}{\}}
\DeclarePairedDelimiter\abs{\lvert}{\rvert}
\DeclarePairedDelimiter\norm{\lVert}{\rVert}

\theoremstyle{remark}
\newtheorem{rem}{\textbf{Remark}}

\title{Counter-example Guided Learning of Bounds on Environment Behavior}
\author{Yuxiao Chen, Sumanth Dathathri, Tung Phan-Minh, and Richard M. Murray\\
 California Institute of Technology\\
  \texttt{\{chenyx, sdathath, tung\}@caltech.edu, murray@cds.caltech.edu} \\}

\date{March 2019}

\begin{document}
\maketitle

\begin{abstract}
There is a growing interest in building autonomous systems that interact with complex environments.
The difficulty associated with obtaining an accurate model for such environments poses a challenge to the task of assessing and guaranteeing the system's performance.
We present a data-driven solution that allows for a system to be evaluated for specification conformance without an accurate model of the environment.
Our approach involves learning a conservative reactive bound of the environment's behavior using data and specification of the system's desired behavior.
First, the approach begins by learning a conservative reactive bound on the environment's actions that captures its possible behaviors with high probability.
This bound is then used to assist verification, and if the verification fails under this bound, the algorithm returns counter-examples to show how failure occurs and then uses these to refine the bound.
We demonstrate the applicability of the approach through two case-studies: i) verifying controllers for a toy multi-robot system, and ii) verifying an instance of human-robot interaction during a lane-change maneuver given real-world human driving data.
\end{abstract}
\section{Introduction}
In control and decision-making tasks, typically, the system can be divided into the controlled agent and the uncontrolled environment, which is the source of exogenous disturbances and uncertainties.
For systems that are safety-critical, given a control policy, it is desirable that we are able to provide guarantees with regard to task fulfillment and safe behavior.
In this regard, formal verification allows us to provide strong guarantees about the absence of unsafe behavior for the controlled agent under all possible behaviors of the environment.
However, to leverage the power of formal verification it is necessary to obtain a reliable model of the environment.
For systems that exhibit complex behavior, accurately modeling the complex environment can be limited by the expressiveness of the model being used, and the amount of data available. In addition, the environment behavior is usually nondeterministic, which may be hard to express.

To overcome this need for explicitly modeling the complete environment, we instead propose an alternative approach that computes \emph{reactive bounds} on the set of feasible behaviors of the environment by leveraging i) a specification for the behavior of the system, and ii) the controller being verified, in addition to the real data collected from naturalistic environment behaviors.
The bounds computed here are \emph{reactive} in that they capture the reactiveness of the environment towards the controlled agent, i.e., the set of possible environment behaviors changes with the scenario described by the system states.
For the purpose of verification, it often may suffice to compute a conservative bound on the possible behaviors of the unmodeled environment as opposed to learning a complete model of the environment.
For example, consider a scenario where two cars are driving perpendicularly towards an intersection with the same distance $d$ to the intersection, and the controlled car (whose controller we seek to verify) drives with a constant velocity $v$.
To verify that there will be no collision (the desired specification), it suffices to bound the velocity of the other car by $u$ such that:
$\frac{d + w}{u}  < \frac{d}{v}$, where $w$ is the width of the intersection.
This guarantees that the uncontrolled car will not enter the intersection before the controlled car leaves the intersection.
Here, a conservative bound on the behavior of the uncontrolled agent (environment) suffices to verify the behavior of the controlled agent, without modeling the exact behavior of the environment.
However, having bounds on the behavior of the environment that are too loose might result in the controller not being provably safe with regard to the bound, despite being safe with regard to the true environment.
In this direction, we present a novel approach where the specification and the controller being verified guide the learning of a reactive bound from data.
%
%

Recently, there has been an increased interest in data-driven verification for cyber-physical systems \cite{DRYVR, Kozarev, 8272345}.
In \cite{Balkan}, a data-driven automated approach is proposed to identify non-converging behaviors in black box control systems.
In \citep{sofiedatadriven}, the authors propose an approach based on Bayesian inference and reachability analysis for verifying the behavior of systems.
However, the approach does not decompose the system into the uncontrolled environment and the controller, and is limited to the model class of linear time-invariant systems.
In contrast, our approach leverages known policies for the controlled agent to enable verification of their behavior with complex environments.
A closely related direction of work is on mining specifications \cite{JinX, vazquez2017logical, NIPS2018Marcell} from data. The mined specifications are often used as task specifications, as opposed to being used to verify a given controller.
For the case of human-robot interaction, treating it as a multi-agent task and leveraging the influence of the autonomous agent on the (uncontrolled) human has been considered before \cite{6630866, 7139555}.
In \cite{sadigh2016planning}, the authors propose an approach that learns a reward function to model the behavior of the uncontrolled agent and then plans for the autonomous agent, leveraging this reward function to model the influence of the autonomous agent on the uncontrolled agent.
Here the authors incorporate the environments behavior into the planning phase, while in contrast, we leverage the controlled agent's policy and the desired safety specification for the system to characterize the environment's behavior.
\citep{Strabala} demonstrates the benefits of learning the intent of the uncontrolled agent prior to the physical event, and leveraging this for seamless collaboration.
In our work, we aim at learning a set of possible environment behaviors that changes with the state, which enables worst-case analysis and is subsequently used for verification of the system.
To summarize, our main contributions are:
\begin{itemize}
    \item A framework for characterizing the behavior of the environment for which the given controller can be certified safe with high probability: given data characterizing the behavior of the environment, the safety-specification and the controller for the system.
    \item When the controller is inherently unsafe, a feedback mechanism for the controller.
    \item Experiments demonstrating the efficiency of the proposed approach on a toy multi-agent task, and verifying a controller for an autonomous vehicle performing a lane change while interacting with a human-driven vehicle characterized by real-world data.
\end{itemize}

The paper is structured as follows. First, we provide a brief introduction to Signal Temporal Logic (STL) and Random Convex Programs (RCP). Then, we provide an overview of the proposed approach, followed by a theoretical analysis. Lastly, we describe our results from two empirical case studies on problems from diverse domains and then conclude.
%

%

\section{Preliminaries}

Consider a dynamical system $\Sigma$ described by differential or difference equations:
\begin{equation}
    x^+ = f(x,u,d),
\end{equation}
where $x\in\mathcal{X}$ is the state, $u\in\mathcal{U}$ is the control input, $d\in\mathcal{D}$ is environment input, and $\mathcal{X}$, $\mathcal{U}$, $\mathcal{D}$ incorporates the physical limits of the variables.
A run of $\Sigma$ is an indexed family $\sigma$ consisting of 3-tuples of the form $\sigma_t = (x(t),u(t),d(t))$, satisfying the dynamics equation.
If $\sigma$ is a run of $\Sigma$, we will also write $\sigma \models\Sigma$.
A run can be infinite or finite with horizon $T$.
\subsection{Signal Temporal Logic}
To express desired properties for the system, we use the formalism of \textit{Signal Temporal Logic} (STL) \cite{donze2013signal}, an extension of Linear Temporal Logic to vector-valued signals.
For any $a,b \in \mathbb{R}$, we will denote by $[a,b]$ the closed interval $\set{x \in \mathbb{R} \mid a \leq x \leq b}$.
STL formulae can be built recursively~as:
\begin{equation*}
\label{stl:syntax}
\varphi \triangleq \texttt{True} \mid p \mid \lnot \varphi \mid \varphi \land \psi \mid \varphi \mathbf{U}_{[a,b]} \psi,
\end{equation*}
where $p$ is an atomic predicate of the form:
$
p \triangleq f(\sigma(t)) > 0
$ for some $f: \mathbb{R}^m \rightarrow \mathbb{R}$.
We write $(\sigma, t) \models \varphi$ to indicate that $\varphi$ holds for $\sigma(t)$.
The satisfaction of a signal $\sigma$ at time $t$ for any of the building block formulae in \eqref{stl:syntax} is defined in the obvious way, except perhaps the one with the ``until'' operator $\textbf{U}$, which is given as:
$$
(\sigma,t) \models \varphi \textbf{U}_{[a,b]} \psi \Leftrightarrow \exists \tau \in [t+a, t+b]. (\sigma,\tau) \models \psi \land \forall \tau' \in [t, \tau]. (\sigma, \tau') \models \varphi.
$$
For convenience, we can define the ``eventually'' $\lozenge$ and ``always'' $\square$ operators as $\lozenge_{[a,b]} \varphi \triangleq \texttt{True} \textbf{U}_{[a,b]} \varphi$ and $\square_{[a,b]} \varphi \triangleq \lnot (\lozenge_{[a,b]} \lnot \varphi)$ such that $(\sigma,t) \models \lozenge_{[a,b]}{\varphi}$ if and only if $\varphi$ is satisfied at least once within a time window of length $b-a$, $a$ time units from $t$ while $(\sigma,t) \models \square_{[a,b]}{\varphi}$ requires that $\varphi$ should always be satisfied within that time window.
\subsection{Random Convex Program}\label{sec:RCP}
The reliability of the proposed approach is based on the theory of random convex programs (RCP).
Let $P[K]$ denote a (minimization) optimization problem with a known objective function and constraint set $K$, and let $\text{Obj}[K]$ denote the optimal objective value of $P[K]$. A constraint $k$ is a supporting constraint if $\text{Obj}[K\backslash \set{k}]<\text{Obj}[K]$.
The setup for an RCP is the following:
\begin{equation}\label{eq:RCP}
  \begin{aligned}
\min \;&J(\alpha)\;\\
\rm{s.t.}& \alpha\in Q(\delta_i), \forall {\delta_1},...,{\delta_N}\;\textrm{i.i.d samples of } \delta,
\end{aligned}
\end{equation}
where $\alpha\in\mathbb{R}^n$, $Q(\delta_i)\subseteq \mathbb{R}^n$ is a convex set determined by $\delta_i$, and $J(\alpha)$ is convex. $\delta\in\Delta$ is a random variable in the space $\Delta$ and $\left\{\delta_i\right\}$ are independently identically distributed samples of $\delta$. Each $\delta_i$ would pose a convex constraint on $\alpha$. If we randomly draw $N$ samples of $\delta$, and denote it as $\omega\doteq\delta_{1:N}\in\Delta^N$, then let $Q(\omega)\doteq\bigcap\nolimits_{i=1}^{N}{Q(\delta_i)}$, define
\begin{equation}\label{eq:V_def}
  V^*(\omega) = \mathbb{P}\left\{\delta\in\Delta:\text{Obj}([Q(\omega),Q(\delta)])>\text{Obj}[Q(\omega)]\right\},
\end{equation}
 which is the probability that an additional sample added on top of $\omega$ would change the objective value of the original optimization with constraints determined by $\omega$. \cite{calafiore2010random} gives upper bound on $\mathbb{P}(V^*(\omega)\ge\epsilon)$ given $1\ge\epsilon>0$ for a randomly drawn sequence of samples $\omega$. First, we recall the following relevant lemma from \cite{calafiore2010random}:
 \begin{lem}\label{lemma:RCP}
  Consider the random convex program in \eqref{eq:RCP} where $\alpha\in\mathbb{R}^n$. When $N\ge \zeta$, $
  \mathbb{P}\left\{\omega\in\Delta^N:V^*(\omega)>\epsilon\right\}\le\Phi(\epsilon,\zeta-1,N)\le\Phi(\epsilon,n,N),
$
where $\zeta$ is the Helly's dimension denoting the maximum number of supporting constraints, which is bounded by $n+1$.
  \begin{equation}\label{eq:Phi}
    \Phi(\epsilon,k,N)=\sum\limits_{j=0}^{k}\binom{N}{j}\epsilon^j(1-\epsilon)^{N-j}
  \end{equation}
is the cumulative distribution of a binomial random variable, that is, the probability of getting no more than $k$ successes in $N$ Bernoulli experiments with success probability $\epsilon$.
 \end{lem}
 This is Theorem 3.3 in \cite{calafiore2010random}, which shows that the result of the RCP is likely to be true for unseen $\delta$ drawn from the same distribution under large $N$ and small $n$. We will revisit this lemma in Section \ref{sec:SVM} to prove probabilistic correctness of the proposed method.

\section{Approach}

Given a dynamical system $\Sigma$, a specification about the initial condition $\varphi_0$, a controller $\varphi_c$ and a performance specification $\varphi_p$, the goal is to verify that $\varphi_p$ is satisfied by all the runs of the system, with high probability, when the control input and initial condition satisfy $\varphi_0$ and $\varphi_c$. However, since the system is interacting with the environment, $\varphi_c$ and $\varphi_0$ alone typically do not imply $\varphi_p$, i.e., the verification fails trivially assuming that the environment can choose arbitrary behaviors. Therefore, we look for an assumption of the environment $\varphi_e$ that is correct (or at least correct with high probability) such that
$\Sigma\wedge\varphi_0\wedge\varphi_e\wedge\varphi_c\Rightarrow\varphi_p$, where the implication is understood as
\begin{equation*}
    \forall (x(t),u(t),d(t)),\left( {(x(t),u(t),d(t)) \models \Sigma \wedge {\varphi _0} \wedge {\varphi _c} \wedge {\varphi _e}} \right) \Rightarrow \left( {(x(t),u(t),d(t)) \models {\varphi _p}} \right).
\end{equation*}

We propose a framework that learns $\varphi_e$ by using a falsification procedure in the loop, in addition to data of interaction between the system and the environment. Figure  \ref{fig:flow_chart} depicts an overview of the framework being proposed. The falsification module takes a fixed controller, the reactive bound of the environment, and the specification $\varphi_p$ as inputs, and either returns traces of the system evolution and the environment behavior that falsify $\varphi_p$ under the given controller, or returns a flag saying that no falsifying trace could be found.
\begin{wrapfigure}{r}{0.55\textwidth}
    \centering
    \includegraphics[scale =0.43]{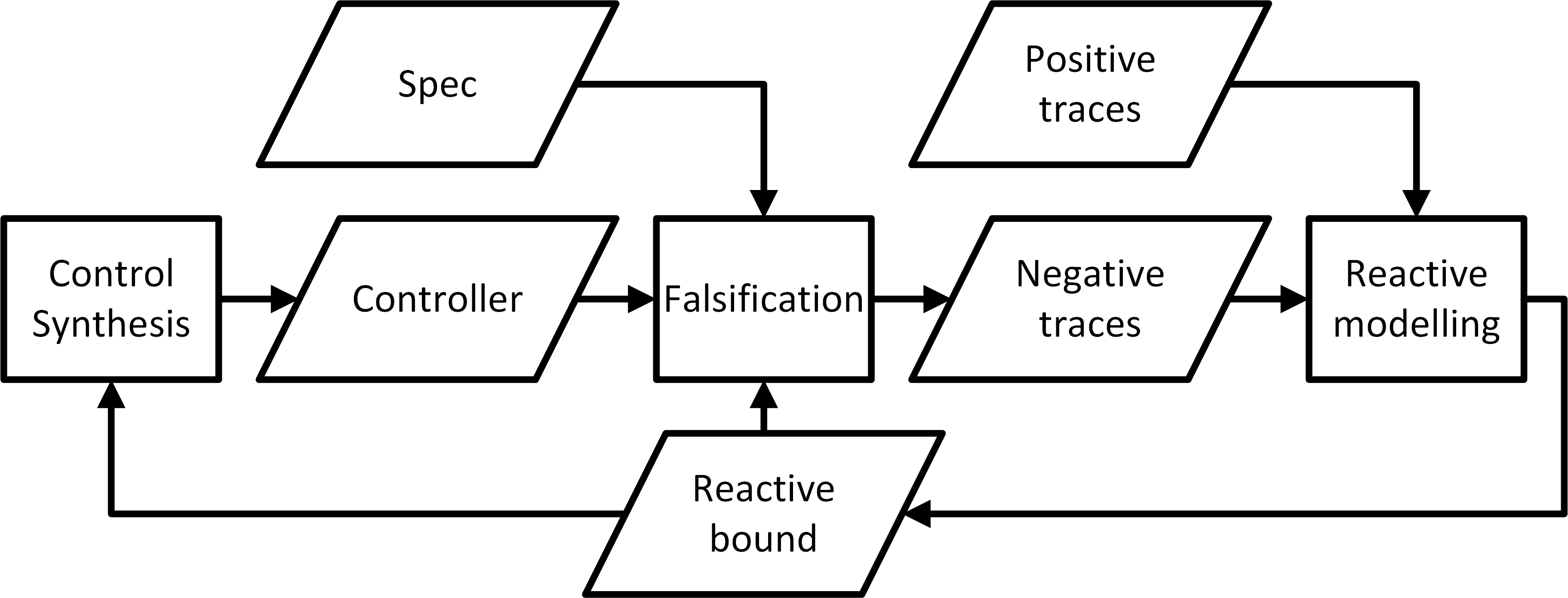}
    \caption{Verification with reactive modelling of the environment}
    \label{fig:flow_chart}
\end{wrapfigure}
In particular, we use the tool S-Taliro as an oracle falsifier, which uses stochastic sampling and can handle STL formulae, see \cite{annpureddy2011s} for details. When the falsification returns no trace, the procedure terminates and the verification is successful.
When falsifying traces are found, they are fed into the reactive modelling module where the positive traces collected from the actual interaction with the environment and the negative traces from the falsification module go through a classification process and the output is a reactive bound that maps the agent's state $x$ to a set of possible behaviors $d$ for the environment, denoted as ${S_d}(x)$. To obtain a bound on the influence of $x$ on the environment response $d$, we want to find a function $h:\mathcal{X}\times\mathcal{D}\to\mathbb{R}$ such that $h(x,d)\geq 0$ indicates that $d$ is possible under $x$, and $h(x,d)<0$ indicates that $d$ is not possible under $x$. The set of possible behaviors for the environment given by the reactive bound can be represented as
\begin{equation}\label{eq:reactive_bound}
S_d(x)=\left\{d\mid h(x,d)\ge0\right\}.
\end{equation}
\begin{rem}
Not every snapshot from the falsifying trace is included in the negative data.
We use an ad-hoc selection scheme to pick out `important' snapshots based on criterion such as the distance between the robots (Section \ref{sec:robot}), and lateral position for the AV (Section \ref{sec:LC}).
\end{rem}
\begin{rem}
Since the data is in the format of snapshots, $\varphi_e$ is limited to the form $\square \varphi$ where $\varphi$ has no dependence on time. One could use a parameterized form for $\varphi_e$ and project the traces to the parameter space, such as in \cite{NIPS2018Marcell}, but this limits  $\varphi_e$ to having only monotonic atoms.
\end{rem}
\subsection{$\mathcal{L}_1$ piecewise SVM for reactive modelling}\label{sec:SVM}
Given the positive and negative traces, we need to learn an indicator function $h$, which then gives rise to the reactive bound. This can be solved as a classification problem. There are numerous classification tools in the literature, such as neural networks and logistic regression. For the reactive modelling problem, in addition to good classification accuracy, the following two requirements are critical: 1) the probability of false negative should be low, even for the unseen data, and 2) the classification result $h$ should have an analytic form for its classification boundary.
For the first requirement, note that $h(x,d)<0$ indicates that $d$ will not happen under $x$, and the verification process will ignore such environment input under $x$.
Therefore, the probability of false negative should be very low to guarantee the correctness of the reactive bound and consequently guaranteeing safety.
The reliability analysis of the proposed approach is based on the theory of random convex program (RCP), which we discuss in detail in Section \ref{sec:RCP_proof}.
For the bound in Lemma \ref{lemma:RCP}, we would like the number of parameters for the classifier to be small -- this prevents overfitting, and enables us to provide better probabilistic guarantees.
The second requirement comes from the fact that the reactive bound will be used explicitly during verification and control synthesis. Therefore, its explicit form should be known.

Due to the two requirements, we choose Support Vector Machines (SVMs) with explicit features as the classification method.
SVMs fit into the setup of Lemma \ref{lemma:RCP} if the positive data is sampled from an i.i.d. distribution.
In particular, we propose an novel expressive $\mathcal{L}_1$ piecewise SVM, which is based on the work on $\mathcal{L}_1$ SVM in \cite{chen2018modelling}.
\cite{chen2018modelling} showed that with a proper cost function, the following optimization solves the $\mathcal{L}_1$ SVM:
\begin{equation}\label{eq:L1_SVM}
\begin{aligned}
\mathop {\min }\limits_{v,c,M} & k^{\intercal}M\\
\rm{s.t.} &\left\|v\right\|\le 1,\;y_i=1\Rightarrow M_i\ge0,\\
&\forall i=1,...,N,~{v^\intercal}{\phi(z_i)}y_i - {M_i} + c y_i = 0,\\
\end{aligned}
\end{equation}
where $z_i\in\mathbb{R}^m$ is the $i$-th data point and $\phi:\mathbb{R}^m\to\mathbb{R}^p$ maps the data to the feature space, the flag $y_i=1$ for positive data points and $y_i=-1$ for negative data points. $v\in\mathbb{R}^p$ is the support vector and $c$ is the offset, therefore the Helly dimension is $p+2$.  $M\in\mathbb{R}^N$ is the slack vector and $k$ is the cost vector with
$k_i>0$ for all $i$.
See section~\ref{app:derivation} for more detail. It is required that all the positive data points are correctly classified (no false negatives), which is needed for the reliability proof.

However, the $\mathcal{L}_1$ SVM suffers from the lack of expressibility, especially for high dimensional $d$.
We propose two improvements on the SVM: 1) Algorithm \ref{alg:SVM_multiple_hyperplanes}, which generates multiple separating hyperplanes and represents the positive data region as a polytope. 2) Based on Algorithm \ref{alg:SVM_multiple_hyperplanes}, we introduce Algorithm \ref{alg:piecewise_SVM}, which allows a piecewise structure for the SVM where a different polytope represents the positive data in each region, and automatically synthesises the piecewise regions.

The original SVM generates one separating hyperplane in the feature space, which results in a reactive bound with a smooth boundary in $\mathcal{X}\times\mathcal{D}$. In order to make the reactive bound more expressive, we propose a piecewise $\mathcal{L}_1$ SVM with multiple separating hyperplanes. SVM with multiple separating hyperplanes is achieved with the following greedy algorithm:
\begin{algorithm}[H]
\centering
    \begin{algorithmic}
            \State  Input: positive data $\phi_p$, negative data $\phi_n$
            \State $\phi^{active}_n \gets \phi_n$
            \For{i=1:$N_h$}
            \State Perform $\mathcal{L}_1$ SVM with $\phi_p$ and $\phi^{active}_n$, get $v^i$, $c^i$, slacks $M_p$, $M_n$
            \State $\phi_n^{active}=\left\{\phi_n^{active}(j)|M_n(j)\ge \epsilon\right\}$
             \EndFor
    \end{algorithmic}
    \caption{$\mathcal{L}_1$ SVM with multiple separating hyperplanes}
    \label{alg:SVM_multiple_hyperplanes}
    \end{algorithm}
$N_h$ is the number of separating hyperplanes, $\phi_n^{active}$ is the set of negative data points that are close to the farthest separating hyperplane, $\epsilon$ is the threshold for picking $\phi_n^{active}$ and $M_p$ and $M_n$ are the slacks for positive and negative data. Each SVM computation generates one hyperplane with $v^i$ and $c^i$ and the indicator function is $h(z)=\mathop{\min}\limits_{i=1,.,N_h}\left\{{(v^i)^\intercal z+c^i}\right\}$.
We can further improve the expressibility by introducing a piecewise structure , which is particularly helpful when the problem itself has a piecewise structure, as demonstrated in Section \ref{sec:robot}. Moreover, we develop an auto-tuning piecewise SVM that adjusts the dividing point automatically based on the data by the use of membership functions.

For clarity, we present the piecewise SVM with 2 regions, but note that it can be easily extended to cases with more than 2 regions. Let $g:\mathbb{R}^n\to\mathbb{R}$ be a scalar function and $\kappa$ be a scalar variable. We will divide the state space by the threshold $g(z)=\kappa$. First, define the membership functions using the sigmoid:
\begin{equation*}
    m_1(z,\kappa)=\frac{1}{1+\exp(\gamma(g(z)-\kappa))}, ~ m_2(z,\kappa)=\frac{\exp(\gamma(g(z)-\kappa))}{1+\exp(\gamma(g(z)-\kappa))},
\end{equation*}
where $\gamma$ is a tuning parameter that controls the steepness of the sigmoid.
Note that $m_1(z,\kappa)+m_2(z,\kappa)=1$. When there are $d>2$ regions, one simply construct $d$ membership functions that are nonnegative and add up to 1. With 2 regions, the original feature is extended to
$\bar{\phi}(z)=[m_1(z,\kappa)\cdot\phi(z);m_2(z,\kappa)\cdot\phi(z)]$. We then perform $\mathcal{L}_1$ SVM with this new feature vector.

Once the SVM is trained, notice that by \eqref{eq:L1_SVM}, $M_i=v^\intercal \bar{\phi}(z_i)y_i+dy_i$, taking derivative of the objective function over $\kappa$, we have:
\begin{equation*}
    \frac{{d({k^\intercal}M)}}{{d\kappa }} = \sum\nolimits_i {{k_i}{y_i}\frac{{d({v^\intercal}\bar \phi ({z_i}))}}{{d\kappa }}}  = \sum\nolimits_i {{k_i}{y_i}\left( {{v_{1:p}}^\intercal\phi ({z_i})\frac{{\partial{m_1}}}{{\partial\kappa }} + {v_{p + 1:2p}}^\intercal\phi ({z_i})\frac{{\partial{m_2}}}{{\partial\kappa }}} \right)} ,
\end{equation*}
\begin{equation*}
    {\left. {\frac{{\partial{m_1}}}{{\partial\kappa }}} \right|_z} =  - {\left. {\frac{{\partial{m_2}}}{{\partial\kappa }}} \right|_z} = \frac{{\gamma {m_1(z,\kappa)}}}{{1 + \exp (\gamma (g(z) - \kappa ))}},
\end{equation*}
so we obtain the analytic form of the gradient of the objective function over $\kappa$.
The overall algorithm alternates between the $\mathcal{L}_1$ SVM and optimizing over $\kappa$ by gradient descent, as in Algorithm~\ref{alg:piecewise_SVM}.
\begin{rem}
The setup for $\mathcal{L}_1$ SVM allows for piecewise cost function of $M$ by spliting $M=M^+ + M^-$ with $M^+\ge 0,M^-\le 0$, see \cite{chen2018modelling} for detail. In the gradient descent step, we maintain the constraint by assigning a large penalty on $M^-$.
\end{rem}

\begin{algorithm}[H]
\centering
    \begin{algorithmic}
            \State  Initialize $\kappa$
            \For{iter=1:T}
            \State Compute membership functions $m_1$, $m_2$
            \State Perform Algorithm \ref{alg:SVM_multiple_hyperplanes} with $\bar{\phi}(z)=[m_1(z)\cdot\phi(z);m_2(z)\cdot\phi(z)]$
            \State Perform gradient descent to optimize $\kappa$
             \EndFor
    \end{algorithmic}
    \caption{Auto-tuning piecewise $\mathcal{L}_1$ SVM}
    \label{alg:piecewise_SVM}
    \end{algorithm}
\vspace{-0.3cm}
\subsection{Reliability analysis with RCP}\label{sec:RCP_proof}
Next, we provide reliability analysis for the reactive bound. For the ordinary $\mathcal{L}_1$ SVM, we have:
\begin{thm}
\label{thm:RCP}
Given a positive data set with $N$ points drawn i.i.d. from a fixed (not necessarily known) distribution, and a negative data set, let $p$ be the dimension of the feature vector $\phi(z)$, $p+2<N$, the $\mathcal{L}_1$ SVM in \eqref{eq:L1_SVM} is always feasible. Denote the solution as $[v,c,M]$, which satisfies $v^\intercal \phi(z_i)+c\ge 0,$ for all positive data points. Then for an unseen data point $z_{N+1}$ from the same distribution, for any $0\le\epsilon\le1$, we have:
\begin{equation}
    \mathbb{P}\left\{    \mathbb{P}\left\{\neg( v^\intercal \phi(z_{N+1})+c>0)\right\}>\epsilon       \right\}\le \sum\nolimits_{j = 0}^{p+2} \binom{N}{j}{{\epsilon^j}{{(1 - \epsilon)}^{N - j}}}.
\end{equation}
\end{thm}
See Appendix \ref{sec:RCP_thm_proof} for proof.
Theorem \ref{thm:RCP} gives an upper bound for the probability of the probability of misclassification for unseen data to be higher than a threshold, which decreases with the size of the dataset $N$ and increases with the feature dimension $p$.
For SVM with $N_h$ hyperplanes, we provide the following corollary.
\begin{cor}
Given the condition in Theorem \ref{thm:RCP}, the $\mathcal{L}_1$ SVM with $N_h$ separating hyperplanes is always feasible and define
$\bar{\epsilon}(k,N)=\mathop{\min}\limits_{0\le\epsilon\le 1}\epsilon+\Phi(\epsilon,k,N)(1-\epsilon)$, where $\Phi$ is defined in \eqref{eq:Phi}, then
for any $0\le\epsilon'\le1$,
\[ \mathbb{P}\left\{    \mathbb{P}\left\{\neg\left(\bigwedge\limits_{i=1}^{N_h} (v^i)^\intercal \phi(z_{N+1})+c^i>0\right)\right\}>\epsilon'\right\}\le \frac{\bar{\epsilon}(N,p+2)N_h}{\epsilon'} \]
\label{thm:multi_hyperplane_RCP}
\end{cor}
See Appendix \ref{sec:RCP_cor_proof} for proof.
\begin{rem}
The auto-tuning piecewise SVM in Algorithm \ref{alg:piecewise_SVM} changes the optimization problem every time it updates $\kappa$, which does not allow us to directly apply Theorem \ref{thm:RCP}.
To overcome this, a simple solution is to separate the positive data points into two batches, using the first batch to find a good separation of the state space, i.e., find a good $\kappa$, and the second batch to obtain the reactive bound while fixing $\kappa$.
The size of the second batch determines the probability of misclassification.
\end{rem}
\section{Case Study}
\vspace{-5pt}
\subsection{Multi-robot navigation}\label{sec:robot}
As a toy example, we consider a multi-robot navigation problem consisting of two robots as shown in Fig. \ref{fig:robots}. We denote the positions of the two robots by $p_1, p_2 \in [-l, l]^2 \subset \mathbb{R}^2$. The robots are characterized by the integrator dynamics
$
\dot{p}_i = v_i,
$
where $i \in \set{1,2}$ and $v_i$ satisfies $\left\|v_i\right\|_2\le v^{\max}$ .
The specification for the system is to always maintain distance i.e., it has to satisfy the speficiation $\square_{[0, T]} \texttt{connected}$, where $T$ is the time horizon for the STL specification and $\texttt{connected}$ is a predicate defined by:
$
\texttt{connected} \triangleq \norm{p_1 - p_2} \leq r^{\max}.
$
Here $r^{\max}$ can be thought of as the maximum communication range.
The red robot $R_1$ is the controlled robot and it simultaneously pursues two objects, a moving target $T_1$ (with bounded velocity), and the blue robot $R_2$. It follows a given controller, which in simulation is set to be:
\begin{equation}\label{eq:R1_con}
    v_1={\bf{sat}}_{{v^{\max }}}({k_1}({p_T} - {p_1}) + {k_2}({p_2} - {p_1})),
\end{equation}
with gains $k_1$ and $k_2$, where $p_T$ denotes the position of the target and
${\bf{sat}}_a(x) = x$ if $ \norm{x} \leq a$ and $a \frac{x}{\norm{x}}$ otherwise.
The motion of the blue robot $R_2$ follows a ``black box'' controller, and we would like to learn an over-approximation of its possible behavior as a function of the state. In particular, we pick a controller with a piecewise structure depending on the location $R_2$:
\begin{equation}
\begin{aligned}\label{eq:R2_con}
    {v_2} = {\bf{sat}}_{v^{\max}}\left( {\begin{bmatrix}
    -0.4& -\beta\\ \beta& -0.4
    \end{bmatrix}(p_2-p_1)} +\Delta v_2\right),
    \end{aligned}
\end{equation}
where $\Delta v_2$ is a bounded random noise and $\beta = 1$ if $p_{2,1}(t) \le 0$  and $-1$ otherwise ($p_{2,1}$ denote the $X$ coordinate of $R_2$). The above controller roughly makes $R_2$ spiral counter-clockwise towards $R_1$ on the left half plane, and spiral clockwise on the right.

To initiate the process, we first collect data with simulation by enforcing \eqref{eq:R1_con} and \eqref{eq:R2_con} on $R_1$ and $R_2$ and let $T_1$ move randomly in the state space. The positive data collected consists of tuples of $[p_1,p_2,v_2]$, which contains information about how $R_2$ moves under different states. Recall \eqref{eq:reactive_bound}, in the two robot case, the state $x$ is $[p_1,p_2]$, and the environment input $d$ is $[v_2,v_T]$, where $v_T=\dot{p}_T$ is the velocity of the target, but we do not explicitly learn a reactive model of $v_T$ and the only constraint for which is the norm bound. In the falsification process, the falsifier can choose $v_2$ and $v_T$ while $v_1$ follows \eqref{eq:R1_con}. When no reactive bound is in place, the only constraint for $v_2$ is the norm bound $v^{\max}$, and the falsifier can easily find falsifying traces. The falsifying traces then generate negative data points with the same structure as the positive data, which is then fed to the reactive modelling module. The reactive modelling module utilizes the auto-tuning piecewise $\mathcal{L}_1$ SVM algorithm introduced in \ref{sec:SVM} with 3 separating hyperplanes and $g(x)=p_{2,1}$ to construct a reactive bound. We choose features that are linear in $v_2$ so that the resulting reactive bound is a polytope $S_{v_2}([p_1,p_2])\in\mathcal{R}^2$, given $[p_1,p_2]$. The reactive bound is then fed to the falsifier, which would project the raw input of $v_2$ to $S_{v_2}([p_1,p_2])$. Since we construct the features for the SVM such that $S_{v_2}$ is a polytope, the projection is easily solved with quadratic programming.

After 5 iterations of updating the reactive bound, the falsifier cannot find a falsifying trace, which means that the controller is verified under the learned reactive bound. Moreover, it turns out that the threshold $\kappa$ converges to $3\times10^{-3}$, which is very close to the actual threshold at $\kappa=0$.
\begin{figure}[H]
\vspace{-0.3cm}
    \centering
    \includegraphics[width=0.18\textwidth]{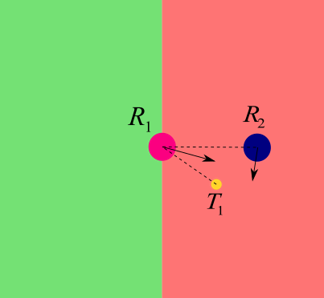}
    \includegraphics[width=0.22\textwidth]{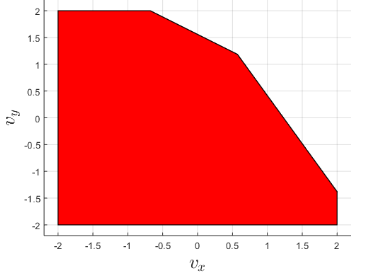}
    \includegraphics[width=0.18\textwidth]{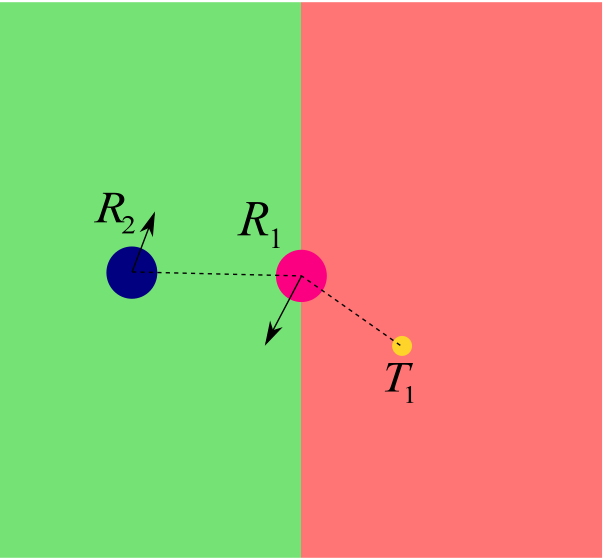}
    \includegraphics[width=0.22\textwidth]{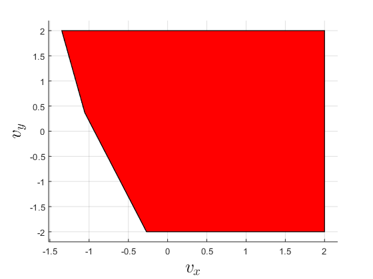}
    \caption{Two robot scenario and Reactive bounds.}
    \label{fig:robots}
    \vspace{-0.4cm}
\end{figure}
\vspace{-5pt}
Fig. \ref{fig:robots} shows two robot positions with corresponding reactive bounds. When $R_2$ is on the right, the reactive bound allows it to spiral clockwise, while the direction of spiralling flips on the left side. But importantly, the worst-case $v_2$, which is to move away from $R_1$ with $v^{\max}$ is not allowed in both cases.
%
%
\vspace{-5pt}
\subsection{Lane Change}\label{sec:LC}
\vspace{-5pt}
A practical application of the proposed method is verification of the lane change control for autonomous driving.
We would like to guarantee with high probability that a given controller can safely finish a lane change within a given horizon.
We consider a scenario as depicted in Fig. \ref{fig:lane_change}, where the autonomous vehicle (AV) attempts to make a lane change with the human driven vehicle (HV) on the back.
\begin{wrapfigure}{r}{0.35\textwidth}
    \centering
    \includegraphics[scale=0.2]{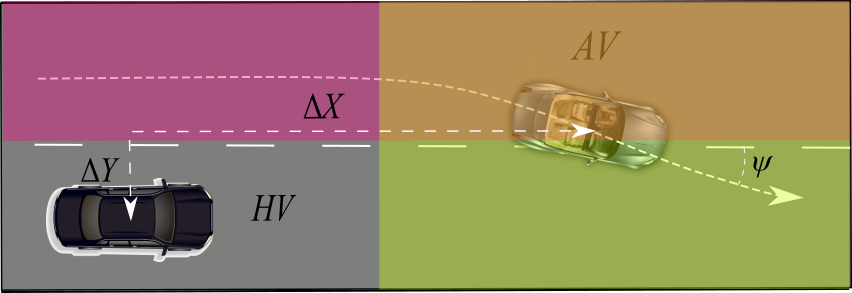}
    \caption{Lane Change Scenario}
    \label{fig:lane_change}
\end{wrapfigure}
The state of the system is
$x = \begin{bmatrix}
    \Delta X & \Delta Y & \Delta v & \psi
    \end{bmatrix}^\intercal$,
    where $\Delta X$ and $\Delta Y$ are the longitudinal and lateral coordinate differences between the two vehicles, $\Delta v$ is the velocity difference and $\psi$ is the heading angle of the AV. The input of AV are the acceration $a_1$ and yaw rate $r_1$, and the input of the HV is the acceleration $a_2$.
    The dynamics is given by:
    \begin{equation}
    \dot{x}=\begin{bmatrix}
    \Delta v&
    v_1 sin(\psi)&
    a_1-a_2&
    r_1
    \end{bmatrix}^\intercal.
\end{equation}

The specification for the problem is to always not collide and keep within the lane, and eventually finish the lane change within horizon $T$. Formally, the specification is expressed in STL:
\begin{equation}\label{eq:LC_spec}
    \square_{[0,T]}(\lnot \mathbf{COL}\wedge\mathbf{LK})\wedge\lozenge_{[0,T]}\mathbf{LC},
\end{equation}

where $\mathbf{COL}$ stands for collision, $\mathbf{LC}$ stands for lane change and $\mathbf{LK}$ stands for lane keeping, which all can be represented as subsets of the state space:
\begin{equation}
\begin{aligned}
    \mathbf{COL} &\Leftrightarrow & |\Delta Y|\le a \wedge |\Delta X|\le b \\
    \mathbf{LC} &\Leftrightarrow & |\Delta Y|\le \epsilon \\
    \mathbf{LK} &\Leftrightarrow & 0.5w-0.5b\ge\Delta Y \ge -1.5w-0.5b,\\
\end{aligned},
\end{equation}
where $a$, $b$ are the length and width of a typical car, $w$ is the width of a lane and $\epsilon \in \mathbb{R}_+$ is a small constant. As an example, we consider a model predictive control scheme with mixed integer programming as the controller for the AV. As shown in Fig. \ref{fig:lane_change}, the AV should stay within the union of the two colored regions within the prediction horizon $T$, which is enforced by the ``big M'' procedure as a mixed integer linear constraint. The MPC controller takes the current value of $a_2$ and assumes an exponential decay within the prediction horizon $a_2(t)=a_2(0)e^{-t/\tau}$, which of course is not accurate, but only a prediction of the future. The lane keeping constraint is also enforced as a linear constraint and the objective function penalizes $\Delta Y$, driving the vehicle to finish the lane change.

The lane change problem was studied in \cite{chen2018modelling}, and we use the same source for positive data, which is from the safety pilot model deployment (SPMD) database with more than 50 million miles of naturalistic driving data \cite{bezzina2015safety}. The feature structure is also inherited from \cite{chen2018modelling}.
Following the procedure shown in Fig. \ref{fig:flow_chart}, the falsification tool starts with simply the physical limit of $a_2$ and tries to falsify the specification in \eqref{eq:LC_spec}, the falsifying traces are then broken into snapshots and treated as negative data. The SVM procedure then generates the reactive bound for $a_2$. In the lane change case, it is not difficult to see that the the safety specification is monotonic w.r.t. $a_2$, i.e., it is always safer for the HV to decelerate. Therefore, the reactive bound for $a_2$ is in the form of an upper bound $a^{\star}_ {\max}(x)$ that changes with the state $x$.

\begin{figure}[H]
    \centering
    \begin{subfigure}[b]{0.22\textwidth}
        \includegraphics[width=\textwidth]{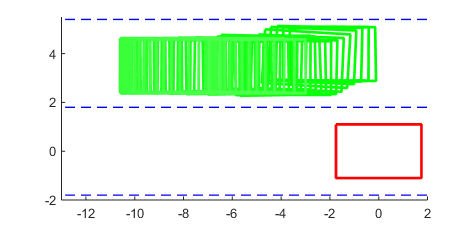}
        \caption{Falsify by blocking}
        \label{fig:LC_fail_block}
    \end{subfigure}
    ~
    \begin{subfigure}[b]{0.22\textwidth}
        \includegraphics[width=\textwidth]{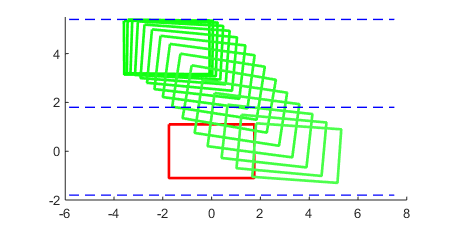}
        \caption{Falsify by collision}
        \label{fig:LC_fail_collision}
    \end{subfigure}
    ~
    \begin{subfigure}[b]{0.4\textwidth}
        \includegraphics[width=\textwidth]{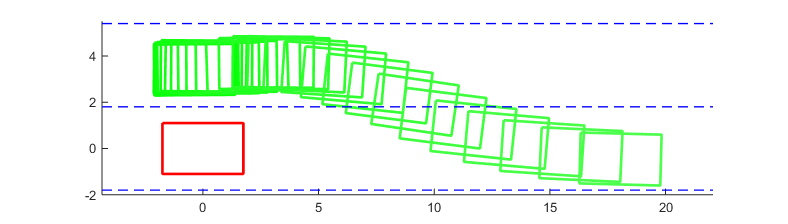}
        \caption{Success run}
        \label{fig:LC_success}
    \end{subfigure}
    \caption{Verification of lane change}\label{fig:LC_verification}
\end{figure}
\vspace{-10pt}
The result of verification for the MPC controller is shown in Fig. \ref{fig:LC_verification}. Without the reactive bound, the falsification procedure is able to falsify the specification by accelerating and blocking the AV from finishing a lane change, as shown in Fig. \ref{fig:LC_fail_block}, and the verification procedure terminates after a maximum iteration number. After 4 iterations, the SVM presented in Section \ref{sec:SVM} generates a reactive bound that makes the falsification infeasible, i.e., verifies that the MPC controller satisfies the specification and a success run is shown in Fig. \ref{fig:LC_success}. However, when we remove the collision avoidance constraint in the MPC controller, the falsification tool finds a falsifying trace by causing a collision with the AV ( Fig. \ref{fig:LC_fail_collision}), thereby providing feedback for the controller design process.

Check \href{github.com/chenyx09/Reactive-modelling}{https://github.com/chenyx09/Reactive-modelling} for the code implementing the proposed method with the two examples.
\\
\textbf{Run-time considerations}
The mean run times over 5 runs are 4977s and 29.6s for the robot problem and for the lane change problem respectively.
\section{Conclusion}

This paper presents a framework that combines falsification and specification learning to learn an over-approximation of the reactive behavior of the environment from real data. There are two key parts of the algorithm, the falsifier and the reactive modelling module. The falsifier can handle specifications written in temporal logic and generate falsifying traces, which is then used by the reactive modelling module together with the positive data to generate the reactive bound. The reliability of the reactive bound is guaranteed by the theory of RCP, which can give probabilistic guarantees determined by the amount of data available. We showed the capability of the proposed framework to handle environment behavior with a piecewise structure as well as demonstrating the result on a practical problem in autonomous driving with real-world human-driving data.
The framework presented here provides a general approach for the diagnosis of autonomous agents interacting with complex environments, such as in the case of human-robot interaction.

\bibliography{mybib}

\section{Appendix}
\subsection{A simple derivation for $\mathcal{L}_1$ SVM}
\label{app:derivation}
Given a data set of features together with their labels $(\phi(z_i), y_i)_{i=1}^N$. Let's say we want to find a hyperplane $\mathcal{H}$ characterized by a normal vector $w$ and an offset $c$ to classify them such that
\begin{enumerate}
    \item \label{req1} all positive data points lie on one side of $\mathcal{H}$
    \item \label{req2} all positive data points lie as close to the hyperplane as possible
    \item \label{req3} negative data points lie as far into or close to the other side of $\mathcal{H}$ as possible
\end{enumerate}
For each $i \in \set{1,2,\ldots,N}$, the distance of $\phi(z_i)$ to the hyperplane $\mathcal{H}$ is equal to $\abs{M_i}$ where $M_i$ is the (unique) solution to the equation
$$
w^T \bigg( \phi(z_i) - M_i \frac{w}{\norm{w}} \bigg) + c  = 0
$$
Then, requiring that all positive data points must all lie on one side of the hyperplane like in requirement~\ref{req1} is equivalent to adding the constraint
\begin{equation}
\label{pos_same}
y_i = 1 \Rightarrow M_i \geq 0
\end{equation}
Actually~\eqref{pos_same} says that all positive data points should lie on the positive side of the hyperplane, but by symmetry this is not a loss of generality.
Consider a linear objective (penalty) function of the form $k^{\intercal} M$ where $M \in \mathbb{R}^N$ has $M_i$ as its $i$-th component. Then requirements \ref{req1}, \ref{req2}, and~\ref{req3} translate to a problem of minimizing $k^{\intercal} M$ with the weight $k$ satisfying $k \succ 0$ (element-wise positive) subject to~\eqref{pos_same}.
\begin{itemize}
    \item For each $i$ such that $y_i = 1$ (and hence $M_i \geq 0$ by requirement \ref{req1} and \eqref{pos_same}), $k_i$ is required to be strictly positive to penalize $\mathcal{H}$ being far away from $\phi(z_i)$ (requirement~\ref{req2}).
    \item For each $i$ such that $y_i = -1$ (negative data points), having $k_i > 0$ will ``force'' $M_i$ to be small (i.e., closer to $-\infty$ on the real axis) to minimize the penalty function and in the process, satisfying requirement~\ref{req3}.
\end{itemize}
A convex relaxation of the above optimization problem is
\begin{equation}\label{eq:simpler_SVM}
\begin{aligned}
\mathop {\min }\limits_{v,d,M} & k^{\intercal}M\\
\rm{s.t.}~&\norm{v} \le 1 \\
& {v^\intercal}{\phi(z_i)} - {M_i} + d = 0 \\
& y_i=1\Rightarrow M_i\ge0
\end{aligned}
\end{equation}
The two optimization problems are equivalent if the optimal solution is negative. Optionally, if we want to weigh correctly or wrongly classified negative data points differently, we can modify the optimization problem to.
\begin{equation}\label{eq:SVM}
\begin{aligned}
\mathop {\min }\limits_{v,d,M} & k_p^{\intercal}M_p + k_{nc}^{\intercal} M_{nc} + k_{nw}^{\intercal} M_{nw}\\
\rm{s.t.}~&\norm{v} \le 1 \\
& {v^\intercal}{\phi(z_i)} - {M_i} + d = 0 \\
& y_i=1\Rightarrow M_i\ge0 \\
& y_i=-1\Rightarrow M_i = M_i^{nw} + M_i^{nc} \land M_i^{nw} \geq 0 \land M_i^{nc} \leq 0
\end{aligned}
\end{equation}
where $k_p, k_{nw}, k_{nc} \succeq 0$ and $k_{nc} \preceq k_{nw}$.

\subsection{Helly's dimension for a variation of RCP}
The following theorem is used in the proof of Theorem \ref{thm:RCP}.
\begin{thm}\label{thm:helly_dim}
  For the following RCP:
  \begin{equation}\label{eq:orig_RCP}
    \begin{aligned}
\mathop {\min }\limits_{\alpha,\beta} &\;J\left( {\alpha,\beta} \right)\\
\rm{s.t.}&f\left( {\alpha,{\delta _i}} \right) \le 0,\forall {\delta _1},...,{\delta _N}\;i.i.d.\\
&g\left( {\alpha,\beta} \right) \le 0.
\end{aligned},
  \end{equation}
  where $\alpha\in\mathbb{R}^{n_1}$, $\beta\in\mathbb{R}^{n_2}$, $f$, $g$ and $J$ convex, the Helly's dimension $\zeta\le n_1+1$.
\end{thm}
\begin{proof}
  Since $J$ is convex, define
  \begin{equation}
    \bar J\left( \alpha \right) = \left\{ {\begin{array}{*{20}{l}}
{\mathop {\min }\limits_\beta \;J\left( {\alpha,\beta} \right)\;\rm{s.t.}\;g\left( {\alpha,\beta} \right) \le 0,\;\;if\;\exists \beta,g\left( {\alpha,\beta} \right) \le 0}\\
{\infty ,\;\;\;\;{\rm{                       otherwise}}}
\end{array}} \right.
  \end{equation}
  Then the RCP in \eqref{eq:orig_RCP} is equivalent to
  \begin{equation}\label{eq:new_RCP}
    \begin{aligned}
\mathop {\min }\limits_\alpha &\;\bar J\left( \alpha \right)\\
\rm{s.t.}&f\left( {\alpha,{\delta _i}} \right) \le 0,\forall {\delta _1},...,{\delta _N}\;i.i.d.
\end{aligned}
  \end{equation}
  When $\bar{J}\neq\infty$, the number of supporting constraint is at most $n_1+1$; when $\bar{J}=\infty$, the number of supporting constraint is zero. This proves that $\zeta\le n_1+1$.
\end{proof}
\subsection{Proof of Theorem \ref{thm:RCP}}\label{sec:RCP_thm_proof}
\begin{proof}
Feasibility can be seen by noticing that the problem is convex and $[w,c,M]=0$ is a solution. Then note that $M$ can be completely eliminated and represented as a function of $[w,c]$ by the equality constraint, and by Theorem \ref{thm:helly_dim}, Helly's dimension is upper bounded by $p+2$. Thus the upper bound on the probability of misclassification for unseen data is obtained by directly using theorem 3.3 in \cite{calafiore2010random}, see the proof therein.
\end{proof}
\subsection{Proof of Corollary \ref{thm:multi_hyperplane_RCP}}\label{sec:RCP_cor_proof}
\begin{proof}
The multi-hyperplane SVM can be viewed as the conjunction of $N_h$ SVMs, therefore we have
\begin{equation}
    \mathbb{P}(\bigvee\limits_{i=1}^{N_h}{((w^i)^\intercal \phi(z_{N+1})+c^i<0)})\le\sum\limits_{i=1}^{N_h}{\mathbb{P}((w^i)^\intercal \phi(z_{N+1})+c^i<0)},
\end{equation}
where for each probability on the right, by Theorem \ref{thm:RCP}, we have
$\mathbb{P}(\mathbb{P}((w^i)^\intercal \phi(z_{N+1})+c^i<0)>\epsilon)<\Phi(\epsilon,p+2,N)$. Since this is true for all $0\le\epsilon\le 1$, we have $\mathbb{E}\{\mathbb{P}((w^i)^\intercal \phi(z_{N+1})+c^i<0)\}\le \mathop{\min}\limits_{0\le\epsilon\le 1}\epsilon+\Phi(\epsilon,p+2,N)(1-\epsilon)$, denoted as $\bar{\epsilon}(N,p+2)$. It is easy to check that this minimum is taken on a bounded function of $\epsilon$ on a compact set, therefore the minimum can always be obtained. Then by Markov inequality, we have for $0\le\epsilon'\le1$,
\[ \mathbb{P}\left\{\sum\limits_{i=1}^{N_h}{\mathbb{P}((w^i)^\intercal \phi(z_{N+1})+c^i<0)}\ge \epsilon'\right\}\le  \frac{\sum\limits_{i=1}^{N_h} \mathbb{E}\{\mathbb{P}((w^i)^\intercal \phi(z_{N+1})+c^i<0)\}}{\epsilon'}\le \frac{\bar{\epsilon}(N,p+2)N_h}{\epsilon'} \]
\end{proof}
\end{document}